%% file: main.tex
\documentclass[11pt]{article}

\usepackage[utf8]{inputenc}
\usepackage[margin=1in,letterpaper]{geometry}

\usepackage{enumitem}
\usepackage[utf8]{inputenc} 
\usepackage[T1]{fontenc}    
\usepackage{hyperref}       
\usepackage{url}            
\usepackage{booktabs}       
\usepackage{microtype}      
\usepackage{graphicx}
\usepackage[fleqn]{amsmath}
\usepackage{amsfonts,amsthm,amssymb}
\usepackage[noend]{algpseudocode}
\usepackage{algorithm}
\usepackage{xcolor}
\usepackage{nicefrac}
\usepackage{tikz}
\usetikzlibrary{calc,positioning,decorations}
\usepackage{caption}
\usepackage{subcaption}
\usepackage{nicefrac}
\include{macros}
\usepackage{pgfplots}
 \newenvironment{proofof}[1]{\noindent{\bf Proof of #1:}}{$\qed$\par}
 
\usepackage{thm-restate}

\usepackage{lmodern}

\newtheorem{lemma}{Lemma}
 \newtheorem{theorem}{Theorem}
 \newtheorem{definition}{Definition}

 \newtheorem{claim}{Claim}
 \newtheorem{fact}{Fact}
 
 \newtheorem{remark}{Remark}

\newcommand{\Vol}{\mathrm{Vol}}
\newcommand{\wt}{\widetilde}

\newcommand{\prr}{\mathrm{pr}}

\usepackage{color}
\definecolor{darkgreen}{rgb}{0,0.5,0}
\usepackage{hyperref}
\hypersetup{
    unicode=false,          
    colorlinks=true,        
    linkcolor=red,          
    citecolor=darkgreen,        
    filecolor=magenta,      
    urlcolor=cyan           
}

\usepackage[capitalize, nameinlink]{cleveref}
\crefname{theorem}{Theorem}{Theorems}
\Crefname{lemma}{Lemma}{Lemmas}
\Crefname{claim}{Claim}{Claims}
\Crefname{fact}{Fact}{Facts}
\Crefname{observation}{Observation}{Observations}

{\makeatletter
 \gdef\xxxmark{%
   \expandafter\ifx\csname @mpargs\endcsname\relax 
     \expandafter\ifx\csname @captype\endcsname\relax 
       \marginpar{xxx}
     \else
       xxx 
     \fi
   \else
     xxx 
   \fi}
 \gdef\xxx{\@ifnextchar[\xxx@lab\xxx@nolab}
 \long\gdef\xxx@lab[#1]#2{{\bf [\xxxmark #2 ---{\sc #1}]}}
 \long\gdef\xxx@nolab#1{{\bf [\xxxmark #1]}}
}

\title{
	Efficient and Local Parallel Random Walks
}

\author{%
	Michael Kapralov\\
	EPFL\\
	\texttt{michael.kapralov@epfl.ch}\\
	\and
	Silvio Lattanzi\\
	Google Research\\
	\texttt{silviol@google.com}\\
	\and
	Navid Nouri\\
	EPFL\\
	\texttt{navid.nouri@epfl.ch}\\
	\and
	Jakab Tardos\\
	EPFL\\
	\texttt{jakab.tardos@epfl.ch}\\
}

\begin{document}

\maketitle

\begin{abstract}
Random walks are a fundamental primitive used in many machine learning algorithms with several 
applications in clustering and semi-supervised learning. Despite their relevance, the first efficient
parallel algorithm to compute random walks has been introduced very recently (\L{}\k{a}cki et al.)
Unfortunately their method has a fundamental shortcoming: their algorithm is non-local in that it heavily relies on computing random walks
out of all nodes in the input graph, even though in many practical applications one is interested in computing random walks
only from a small subset of nodes in the graph. In this paper, we present a new algorithm that overcomes 
this limitation by building random walk efficiently and locally at the same time. We show that our technique
is both memory and round efficient, and in particular yields an efficient parallel local clustering algorithm. Finally, we complement our theoretical analysis with
experimental results showing that our algorithm is significantly more scalable than previous approaches.
\end{abstract}

\section{Introduction}
Random walks are key components of many machine learning algorithms with applications in computing
graph partitioning~\cite{spielman2004nearly, gluch2021spectral}, spectral embeddings~\cite{czumaj2015testing, chiplunkar2018testing},
or network inference~\cite{hoskins2018inferring}, as well as
learning image segmentation~\cite{meila2000learning}, ranking nodes in a graph~\cite{agarwal2007learning}
and many other applications. With the increasing availability and importance of large scale datasets it is
important to design efficient algorithms to compute random walks in large networks. 

Several algorithms for computing random walks in parallel and streaming
models have been proposed in the literature. In the streaming setting, Sarma, Gollapudi and Panigrahy~\cite{sarma2011estimating} introduced
multi-pass streaming algorithms for simulating random walks, and recently Jin~\cite{DBLP:conf/innovations/Jin19} 
gave algorithms for generating a single random walk from a prespecified vertex in one pass. The first efficient
parallel algorithms for this problem have been introduced in the PRAM model~\cite{karger1992fast, halperin1996optimal}.

In a more recent line of work, Bahmani, Chakrabarti, and 
Xin~\cite{bahmani2011fast} designed a parallel algorithm that constructs a single random
walk of length $\ell$ from every node in $O(\log \ell)$ rounds in the massively parallel computation model (MPC), with the important caveat that these walks are not independent (an important property in many applications).  This was followed by the work of Assadi, Sun and Weinstein~\cite{assadi2019massively}, which gave an MPC algorithm for generating random walks in an undirected regular graph. Finally, \L{}\k{a}cki et al.~\cite{lkacki2020walking} presented a new
algorithm to compute random walks of length $\ell$ from every node in an arbitrary undirected graph. The algorithm of~\cite{lkacki2020walking} still uses only $O(\log \ell)$ parallel rounds, and walks computed are now independent.

From a high level perspective, the main idea behind all the MPC algorithms presented in~\cite{bahmani2011fast, assadi2019massively,
lkacki2020walking} is to compute random walks of length $\ell$ by stitching together random walks of length
$\nicefrac{\ell}{2}$ in a single parallel round. The walks of length $\nicefrac{\ell}{2}$ are computed by stitching 
together random walks of length $\nicefrac{\ell}{4}$ and so on. It is possible to prove that such strategy
leads to algorithms that run in $O(\log \ell)$ parallel rounds as shown in previous work (this is also optimal under
the 1-vs-2 cycle conjecture, as shown in~\cite{lkacki2020walking}). Note that this technique in order to succeed
computes in round $i$ several random walks of length $2^i$ for all the nodes in the network in parallel. 
This technique is very effective if we are interested in computing random walks from all the
nodes in the graph, or, more precisely, when the number of walks computed out of a node is proportional to its stationary distribution. However, this approach leads to significant inefficiencies when we are interested in computing random walks only out of a subset of
nodes or for a single node in the graph. This is even more important when we consider that in many applications
as in clustering~\cite{gargi2011large, gleich2012vertex, whang2013overlapping} we are interested in running random
walks only from a small subset of seed nodes. This leads to the natural question: 
\emph{Is it possible to compute efficiently and in parallel random walks only from a subset of nodes in a graph?}

In this paper we answer this question in the affirmative, and we show an application of such a result in local clustering. Before
describing our results in detail, we discuss the precise model of parallelism that we use in this work.

\textbf{The MPC model.} We design algorithms in the massively parallel computation (MPC) model, which is a
theoretical abstraction of real-world system, such as MapReduce \cite{dean2008mapreduce}, Hadoop \cite{white2012hadoop}, 
Spark \cite{zaharia2010spark} and Dryad \cite{isard2007dryad}. The MPC model~\cite{karloff2010model, goodrich2011sorting, beame2013communication}
 is the de-facto standard for analyzing algorithms  for large-scale parallel computing.

Computation in MPC is divided into synchronous \emph{rounds} over multiple machines. Each machine has memory $S$
and at the beginning data is partitioned arbitrarily across machines. During each round, machines process data locally
and then exchange data with the restriction that no machine receives more than $S$ bits of data. The efficiency of an algorithm in this model is 
measured by the number of rounds it takes for the algorithm to terminate, by the size of the memory of every machine and by
the total memory used in the computation. In this paper we focus on designing algorithm in the most restrictive and realistic regime
where $S \in O(n^{\delta})$ for a small constant $\delta\in (0, 1)$ -- these algorithms are called fully scalable.

\textbf{Our contributions.} Our first contribution is an efficient algorithm for computing multiple random walks from a single
node in a graph efficiently. 

\begin{restatable}{theorem}{main}
\label{thm:main}
There exists a fully scalable MPC algorithm that, given a graph $G=(V,E)$ with $n$ vertices and $m$ edges, a root vertex $r$, and parameters $B^*$, $\ell$ and $\lambda$, can simulate $B^*$ independent random walks on $G$ from $r$ of length $\ell$ with an arbitrarily low error, in $O(\log\ell\log_{\lambda} B^*)$ rounds and $\wt O(m\lambda\ell^4+B^*\lambda\ell)$ total space.
\end{restatable}

Our algorithm also applies to the more general problem of generating independent random walks from a subset of nodes in the graph:
\begin{restatable}{theorem}{general}\label{thm:general}
There exists a fully scalable MPC algorithm that, given a graph $G=(V,E)$ with $n$ vertices and $m$ edges and a collection of non-negative integer budgets $(b_u)_{u\in V}$ for vertices in $G$ such that $\sum_{u\in V} b_u=B^*$, parameters $\ell$ and $\lambda$, can simulate, for every $u\in V$,  $b_u$ independent random walks on $G$ of length $\ell$ from $u$ with an arbitrarily low error, in $O(\log\ell\log_{\lambda} B^*)$ rounds and $\wt O(m\lambda\ell^4+B^*\lambda\ell)$ total space. The generated walks are independent across starting vertices $u\in V$.
\end{restatable}

The following remark clarifies the effect of parameter $\lambda$ on the number of machines.
\begin{remark}
	The parameter $\lambda$ has nothing to do with the input of the algorithm, but is a trade-off parameter between space and round complexity. It is useful to think of it as $\lambda=n^\epsilon$ for some $\epsilon$ (not necessarily a constant), in which case we get a round complexity of $O(\log\ell\log B^*/(\epsilon\log n))\le O(\log\ell/\epsilon)$ and a total memory of $\wt O(mn^\epsilon\ell^4 + B^*n^\epsilon\ell)$. We can set $\epsilon$ to, for example, $1/\log\log n$, to get nearly optimal total space and $\wt O(\log\ell)$ round complexity.
	
\end{remark}

If we compare our results with previous works, our algorithm computes truly independent random walks as~\cite{lkacki2020walking} does.
This is in contrast with the algorithm of~\cite{bahmani2011fast}, which introduces dependent constructs not independent walks. 
Our algorithm has significantly better total memory than~\cite{lkacki2020walking}, which would result in memory $\Omega( m\cdot B^*)$ for generating $B^*$ walks out of a root node $r$. This comes at the cost of a slightly higher number of rounds ($\log_{\lambda} B^*$, a factor that in
many applications can be considered constant).

The main idea is to preform multiple cycles of stitching algorithms, changing the initial distribution of the random walks adaptively. More precisely, in an initial cycle we construct only a few random walks, distributed according to the stationary distribution -- this is known to be doable from previous work. Then, in each cycle we increase the number of walks that we build for node $r$ by a factor of
$\lambda$ and we construct the walks only by activating other nodes in the graph that contribute actively in the construction of the random walks for $r$. In
this way we obtain an algorithm that is significantly more work efficient in terms of total memory, compared with previous work.

Our second contribution is to present an application of our algorithm to estimating Personalized PageRank and to local graph clustering. To the best of our knowledge, our algorithm is the first local clustering algorithm that uses a number of parallel rounds that only have a
logarithmic dependency on the length of the random walk used by the local clustering algorithm.

\begin{restatable}{theorem}{final}\label{thm:final}
	For any $\lambda>1$, let $\alpha \in (0,1]$ be a constant and let $C$ be a set satisfying that the conductance of $C$, $\Phi(C)$, is at most $\alpha/10$ and $\Vol(C) \le \frac{2}{3}\Vol(G)$.
	Then there is an MPC algorithm for local clustering that uses $O(\log\ell\cdot\log_\lambda B^*)=O(\log\log n\cdot \log_\lambda (\Vol(C)))$ rounds of communication and total memory
	$\wt{O}(m\lambda\ell^4+B^*\lambda\ell)=\wt{O}(m\lambda + \lambda\Vol(C)^2),$  where $B^* := \frac{10^6\log^3 n}{\eta^2 \alpha^2}$, $\ell:=\frac{10\log n}{\alpha}$ and $\eta = \frac{1}{10\Vol(C)}$, and outputs a cluster with conductance $O(\sqrt{\alpha\log(\Vol(C))})$.
\end{restatable}

Finally we present an experimental analysis of our results where we show that our algorithm to compute random walk is significantly more efficient
than previous work~\cite{lkacki2020walking}, and that our technique scale to very large graphs.

\textbf{Additional related works.} Efficient parallel random walks algorithm have also been presented in distributed computing~\cite{das2013distributed} and using multicores~\cite{DBLP:journals/pvldb/ShunRFM16}. Although the algorithms in~\cite{das2013distributed} require a number of rounds linear in the diameter of the graph. The results in~\cite{DBLP:journals/pvldb/ShunRFM16} are closer in spirit to our work here but from an algorithmic perspective the challenges in developing algorithms in multicore and MPC settings are quite different. In our setting, most of the difficulty is in the fact that there is no shared memory and coordination between machines. As a result, bounding communication between machines and number of rounds is the main focus of this line of research. From an experimental perspective an advantage of the MPC environment is that it can scale to larger instances, in contrast an advantage of the multicore approach is that it is usually faster in practice for medium size instances. In our case study, where one is interested in computing several local clusters from multiple nodes(for example to detect sybil attack(look at~\cite{DBLP:journals/im/AlvisiCELP14} and following work for applications) the MPC approach is often more suitable. This is due to the fact that the computation of multiple Personalized PageRank vectors at the same time often requires a large amount of space.

Several parallel algorithm have also been presented for estimating
PageRank or PersonalizedPageRank~\cite{sarma2011estimating, sarma2015fast, bahmani2011fast}, but those algorithms either have higher round
complexity or introduce dependencies between the PersonalizedPageRank vectors computed for different nodes. Finally there has been also some work in parallelize
local clustering algorithm~\cite{chung2015distributed}, although all previously known methods have complexity linear in the number of steps executed by the random walk/process used in the algorithm(in fact, our method could potentially be used to speed-up this work as well).

\textbf{Notation.} We work on undirected unweighted graphs, which we usually denote by $G=(V,E)$, where $V$ and $E$ are the set of vertices and the set of edges, respectively. We also have $n:=|V|$ and $m:=|E|$, unless otherwise stated. We define matrix $D$ as the diagonal matrix of degrees, i.e., $D_{i,i}=d(v_i)$. Also, we let $A$ be the adjacency matrix, where $A_{i,j}=1$ if and only if there is an edge joining $v_i$ and $v_j$, and $A_{i,j}=0$, otherwise.


\section{MPC random walks}
\label{sec:main}

In this section, we present our main result to compute $B^*$ random walks from a single root vertex $r$ up to a length of $\ell$, then we generalize it to multiple sources. As mentioned before, our main idea is to carefully stitch random walks adaptively to activate nodes locally. In the rest of the section, we start by presenting our main theorem and giving a brief overview of our algorithm. We then describe our stitching algorithm, and analyze the number of random walks we must start from each node so that our algorithms work. Finally, we present the extension of our result to the setting with multiple sources.

\main*
\subsection{Overview of our Algorithm}

Here, we explain the frameworks of stitching and budgeting, which are the two key tools making up our algorithm. For simplicity and without loss of generality, we assume that each vertex $v$ has its own machine that stores its neighborhood, its budgets, and its corresponding random walks.\footnote{In reality multiple vertices may share a machine, if they have low degree; or if a vertex has a high degree it may be accommodated by multiple machines in a constant-depth tree structure.}

\begin{remark}\label{rem:rounding}
	For ease of notation we assume that $\ell = 2^{j}$ for some integer $j$. One can see that this assumption is without loss of generality, because otherwise one can always round $\ell$ to the smallest power of $2$ greater than $\ell$, and solve the problem using the rounded $\ell$. This affects the complexity bounds by at most a constant factor.  
\end{remark}

\textbf{Stitching.}
Here, we explain the framework of stitching, which is a key tool for our algorithm.
At each point in time the machine corresponding to $v$ stores sets of random walks of certain lengths, each starting in $v$. Each edge of each walk is labeled by a number from $1$ to $\ell$, denoting the position it will hold in the completed walk. Thus, a walk of length $s$ could be labeled $(k+1,\ldots,k+s)$ for some $k$. Initially each vertex generates a pre-determined number of random edges (or walks of length one) with each label from $1$ to $\ell$. Thus at this point, we would find walks labeled $1$, $2$, $3$, $\ldots$ on each machine. After the first round of stitching, these will be paired up into walks of length two, and so we will see walks labeled by $(1,2)$, $(3,4)$, $(5,6)$, $\ldots$ on each machine. After the second round of stitching we will see walks of length $4$, such as $(1,2,3,4)$, and so on. Finally, after the last round of stitching, each machine will contain some number of walks of length $\ell$ (labeled from $1$ to $\ell$), as desired.

At any given time let $W_k(v)$ denote the set of walks stored by $v$ whose first label is $k$ and $B(v,k)$ denotes their cardinality -- in the future, we will refer to the function $B$ as the budget. After the initial round of edge generation, $W_k(v)$ consists of $B(v,k)$ individual edges adjacent to $v$, for each $v$ and $k$.

The rounds of communication proceed as follows: in the first round of stitching, for each edge (or length one walk) $e$ in $W_k(v)$, for any {\it odd} $k$, $v$ sends a request for the continuation of the walk to $z$, where $z$ is the other endpoint of $e$. That is, $v$ sends a request to $z$ for an element of $W_{k+1}(z)$. Each vertex sends out all such requests simultaneously in a single MPC round. Following this each vertex replies to each request by sending a walk from the appropriate set. Crucially, each request must be answered with a {\it different, independent} walk. If the number of requests for $W_{k+1}(z)$ exceeds $|W_{k+1}(z)|=B(z,k+1)$, the vertex $z$ declares failure and the algorithm terminates. Otherwise, all such requests are satisfied simultaneously in a single MPC round. Finally, each vertex $v$ increases the length of each of its walks in $W_k(v)$ to two when $k$ is odd, and deletes all remaining walks in $W_k(v)$ when $k$ is even (see \Cref{fig:s}). For a more formal definition see \Cref{alg:stitch}.

\textbf{Budgeting.} A crucial aspect of stitching is that the budgets $B(v,k)$ need to be carefully prescribed. If at any point in time a vertex receives more requests than it can serve, the entire algorithm fails. In the case where $B(\cdot, 1)$ follows the stationary distribution, this can be done (see for instance ~\cite{lkacki2020walking}), since the number of requests -- at least in expectation -- will also follow the stationary distribution. In our setting however, when $B(\cdot,1)$ follows the indicator distribution of $r$, this is much more difficult. We should assign higher budgets to vertices nearer $r$; however, we have no knowledge of which vertices these are.

In other words, the main challenge in making stitching work with low space and few rounds is to set the vertex budgets ($B(v,k)$) accurately enough for stitching to succeed -- this is the main technical contribution of this paper.

Our technique is to run multiple cycles of stitching sequentially. In the first cycle, we simply start from the stationary distribution. Then, with each cycle, we shift closer and closer to the desired distribution, in which the budget of $r$ is much greater than the budgets of other vertices. We do this by augmenting $B(r,1)$ by some parameter $\lambda$ each cycle. This forces us to augment other budgets as well: For example, for $u$ in the neighborhood of $r$ we expect to have a significantly increased budget $B(u,2)$. In order to estimate the demand on $u$ (and all other vertices) we use data from the previous cycle.

Though initially only a few walks simulated by our algorithm start in $r$ (we call these rooted walks), we are still able to derive some information from them. For instance, we can count the number of walks starting in $r$ and reaching $u$ as their second step. If $\kappa$ rooted walks visited $u$ as their second step in the previous cycle, we expect this number to increase to $\lambda\cdot\kappa$ in the following cycle. Hence, we can preemptively increase $B(u,2)$ to approximately $\lambda\cdot\kappa$.

More precisely, we set the initial budget of each vertex to $\sim B_0\cdot\text{deg}(v)$ -- an appropriately scaled version of the stationary distribution. This guarantees that the first round of stitching succeeds. Afterwards we set each budget $B(v,k)$ individually based on the information gathered from the previous cycle. We first count the number of rooted walks that ended up at $v$ as their $k^\text{th}$ step (\Cref{line:kappa-def}). If this number is sufficiently large to be statistically significant (above some carefully chosen threshold $\theta$ in our case, \Cref{line:threshold}), then we estimate the budget $B(v,k)$ to be approximately $\lambda\cdot\kappa$ in the following cycle (\Cref{line:option-1}). On the other hand, if $\kappa$ is deemed too small, it means that rooted random walks rarely reach $v$ as their $k^\text{th}$ step, and the budget $B(v,k)$ remains what it was before.



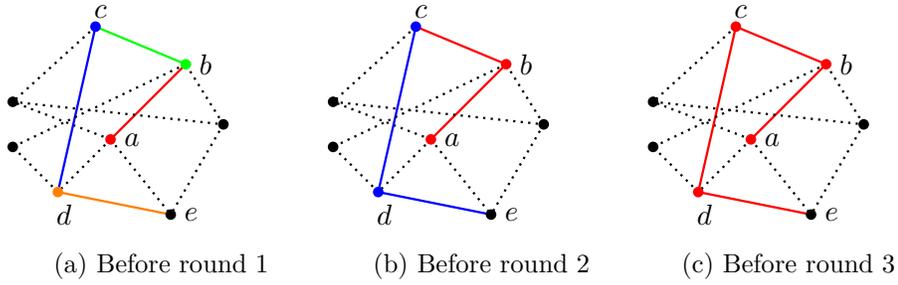
\begin{figure}
	\centering
	
	\begin{subfigure}[b]{0.25\textwidth}
		\begin{tikzpicture}[line width=0.3mm]
			\draw[red] (0,0)--(1,1);
			\draw[dotted] (0,0)--(0.8,-1);
			\draw[dotted] (0,0)--(-0.7,-0.7);
			\draw[dotted] (0,0)--(-1.3,+0.5);
			\draw[dotted] (-0.2,1.5)--(-1.3,+0.5);
			\draw[green] (-0.2,1.5)--(1,1);
			\draw[dotted] (1.5,0.2)--(1,1);
			\draw[dotted] (1.5,0.2)--(0.8,-1);
			\draw[orange] (0.8,-1)--(-0.7,-0.7);
			\draw[blue] (-.7,-0.7)--(-0.2,1.5);
			\draw[dotted] (-1.3,0.5)--(1.5,0.2);
			\draw[dotted] (-0.7,-0.7)--(-1.3,-0.1);
			\draw[dotted] (1,1)--(-1.3,-0.1);
			
			\fill [color=red] (0,0) circle (2pt);	
			\fill [color=green] (1,1) circle (2pt);	
			\fill [color=orange] (-.7,-0.7) circle (2pt);
			\fill [color=black] (0.8,-1) circle (2pt);	
			\fill [color=black] (-1.3,0.5) circle (2pt);
			\fill [color=blue] (-.2,1.5) circle (2pt);	
			\fill [color=black] (1.5,0.2) circle (2pt);		
			\fill [color=black] (-1.3,-0.1) circle (2pt);
			
			\node[text width = 1] at (0.2,0) {$a$};
			\node[text width = 1] at (1.2,1) {$b$};
			\node[text width = 1] at (-0.2,1.7) {$c$};
			\node[text width = 1] at (-.7,-1) {$d$};
			\node[text width = 1] at (1,-1) {$e$};
		\end{tikzpicture}
		\subcaption{Before round 1}\label{subfig:1}
	\end{subfigure}
	\begin{subfigure}[b]{0.25\textwidth}
		\begin{tikzpicture}[line width=0.3mm]
			\draw[red] (0,0)--(1,1);
			\draw[dotted] (0,0)--(0.8,-1);
			\draw[dotted] (0,0)--(-0.7,-0.7);
			\draw[dotted] (0,0)--(-1.3,+0.5);
			\draw[dotted] (-0.2,1.5)--(-1.3,+0.5);
			\draw[red] (-0.2,1.5)--(1,1);
			\draw[dotted] (1.5,0.2)--(1,1);
			\draw[dotted] (1.5,0.2)--(0.8,-1);
			\draw[blue] (0.8,-1)--(-0.7,-0.7);
			\draw[blue] (-.7,-0.7)--(-0.2,1.5);
			\draw[dotted] (-1.3,0.5)--(1.5,0.2);
			\draw[dotted] (-0.7,-0.7)--(-1.3,-0.1);
			\draw[dotted] (1,1)--(-1.3,-0.1);
			
			\fill [color=red] (0,0) circle (2pt);	
			\fill [color=red] (1,1) circle (2pt);	
			\fill [color=blue] (-.7,-0.7) circle (2pt);
			\fill [color=black] (0.8,-1) circle (2pt);	
			\fill [color=black] (-1.3,0.5) circle (2pt);
			\fill [color=blue] (-.2,1.5) circle (2pt);	
			\fill [color=black] (1.5,0.2) circle (2pt);		
			\fill [color=black] (-1.3,-0.1) circle (2pt);
			
			\node[text width = 1] at (0.2,0) {$a$};
			\node[text width = 1] at (1.2,1) {$b$};
			\node[text width = 1] at (-0.2,1.7) {$c$};
			\node[text width = 1] at (-.7,-1) {$d$};
			\node[text width = 1] at (1,-1) {$e$};
		\end{tikzpicture}
		\subcaption{Before round 2}\label{subfig:2}
	\end{subfigure}
	\begin{subfigure}[b]{0.3\textwidth}
		\begin{tikzpicture}[line width=0.3mm]
			\draw[red] (0,0)--(1,1);
			\draw[dotted] (0,0)--(0.8,-1);
			\draw[dotted] (0,0)--(-0.7,-0.7);
			\draw[dotted] (0,0)--(-1.3,+0.5);
			\draw[dotted] (-0.2,1.5)--(-1.3,+0.5);
			\draw[red] (-0.2,1.5)--(1,1);
			\draw[dotted] (1.5,0.2)--(1,1);
			\draw[dotted] (1.5,0.2)--(0.8,-1);
			\draw[red] (0.8,-1)--(-0.7,-0.7);
			\draw[red] (-.7,-0.7)--(-0.2,1.5);
			\draw[dotted] (-1.3,0.5)--(1.5,0.2);
			\draw[dotted] (-0.7,-0.7)--(-1.3,-0.1);
			\draw[dotted] (1,1)--(-1.3,-0.1);
			
			\fill [color=red] (0,0) circle (2pt);	
			\fill [color=red] (1,1) circle (2pt);	
			\fill [color=red] (-.7,-0.7) circle (2pt);
			\fill [color=black] (0.8,-1) circle (2pt);	
			\fill [color=black] (-1.3,0.5) circle (2pt);
			\fill [color=red] (-.2,1.5) circle (2pt);	
			\fill [color=black] (1.5,0.2) circle (2pt);		
			\fill [color=black] (-1.3,-0.1) circle (2pt);
			
			\node[text width = 1] at (0.2,0) {$a$};
			\node[text width = 1] at (1.2,1) {$b$};
			\node[text width = 1] at (-0.2,1.7) {$c$};
			\node[text width = 1] at (-.7,-1) {$d$};
				\node[text width = 1] at (1,-1) {$e$};

		\end{tikzpicture}
		\subcaption{Before round 3~~~~~~~~~~~}\label{subfig:3}
	\end{subfigure}

	\caption{Illustration of stitching algorithm for walk $(a,b,c,d,e)$. The red, green, blue and orange walks in each figure correspond to walks in $W_1(a)$, $W_2(b)$, $W_3(c)$ and $W_4(d)$, respectively.}\label{fig:s}
\end{figure}

\begin{figure}[H]\label{figure2}
	\begin{center}
		\begin{tikzpicture}[scale=0.8]
		\begin{axis}[
		restrict y to domain=-10:10,
		samples=1000,
		width=10cm, height=10cm,
		ymin=-0.2 ,ymax=3.5,
		xmin=-5.8, xmax=5.8,
		xtick={0},
		xticklabels={root},
		ytick={0.5,3},
		yticklabels={$B_0$,$B^*$},
		xlabel={Vertices},
		ylabel={Budgets},
		axis x line=center,
		axis y line=left,
		every axis x label/.style={
			at={(ticklabel* cs:1)},
			anchor=south,}
		]
		
		\addplot [
		domain=-10:10, 
		samples=100, 
		color=blue, thick
		]
		{ 0.5  };
		\addplot [
		domain=-10:10, 
		samples=100, 
		color=blue, thick
		]
		{ max(0.5*2^(-2*x*x)+0.5,0.5)   };
		
		\addplot [
		domain=-10:10, 
		samples=100, 
		color=blue, thick
		]
		{ max(0.5*1.5*2^(-2*x*x)+0.5,0.5)   };
		
		\addplot [
		domain=-10:10, 
		samples=100, 
		color=blue, thick
		]
		{ max(0.5*1.5*1.5*2^(-2*x*x)+0.5,0.5)  };

		\addplot [
		domain=-10:10, 
		samples=100, 
		color=blue, thick
		]
		{ max(0.5*1.5*1.5*1.5*2^(-2*x*x)+0.5,0.5)  };

		\addplot [
		domain=-10:10, 
		samples=100, 
		color=red, thick
		]
		{ max(3,0.5)  };
		
		\node[anchor=west] (source1) at (axis cs:2,1){Initial budgets};
		\draw[->](source1)--(axis cs:0.5,0.55);
		
		\node[anchor=west] (source2) at (axis cs:0.8,1.5){Budgets in iteration $j$};
		\draw[->](source2)--(axis cs:0.5,1.1);
		
		\node[anchor=west] (source3) at (axis cs:0.6,2){Budgets in iteration $j+1$};
		\draw[->](source3)--(axis cs:0.5,1.4);

		\node[anchor=west] (source3) at (axis cs:0.4,2.5){Budgets using \cite{lkacki2020walking}};
		\draw[->](source3)--(axis cs:0.5,2.9);
		
		\end{axis}
		\end{tikzpicture}
		\caption{Illustration of the evolution of $B(v,k)$ for all $v\in V$ for a fixed $k$ on a line graph over the iterations of \Cref{alg:main} and comparing to budgets if one uses \cite{lkacki2020walking}. Vertices are sorted by their order on the line and root is the middle vertex on the line. }	\label{fig:budget}
	\end{center}
	
\end{figure}
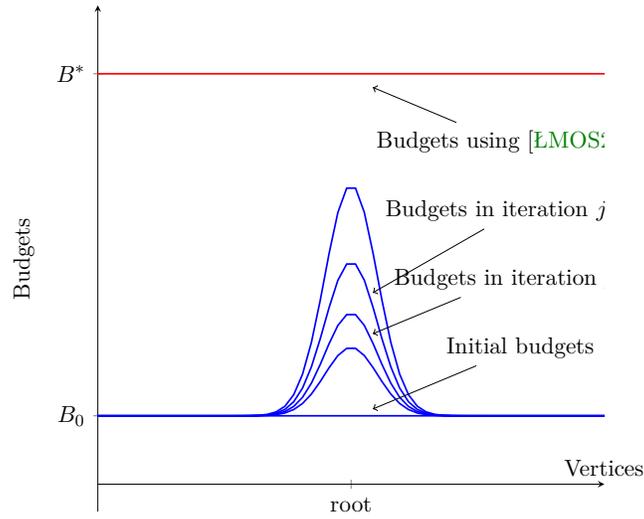

\begin{algorithm}[h!]
	\caption{Stitching algorithm}  
	\label{alg:stitch}
	\begin{algorithmic}[1]
		\Procedure{Stitch}{$G,B$}
		\For{$v\in V$ in parallel}
		\For{$k\in[\ell]$}
		\State $W_k(v)\gets$ a set of $B(v,k)$ independent uniformly random edges adjacent on $v$
		\EndFor
		\EndFor
		\For{$j=1\ldots\log_2\ell$}\Comment{See \Cref{rem:rounding}}
		\For{$v\in V$ in parallel}
		\For{$k\equiv1\ (mod\ 2^j)$}
		\For{walk $p\in W_k(v)$}
		\State {\bf send} $(v,k)$ {\bf to} $z$, where $z\gets$ end vertex of $p$
		\EndFor
		\EndFor
		\EndFor
		\For{$z\in V$ in parallel}
		\For{each message $(v,k)$}
		\If{$W_{k+2^{j-1}}(z)=\emptyset$}
		\State \Return {\bf Fail}
		\EndIf
		\State {\bf send} $(q,k,z)$ {\bf to} $v$, where the walk $q\gets$ a randomly chosen walk in $W_{k+2^{j-1}}(v)$
		\State $W_{k+2^{j-1}}(z)\gets W_{k+2^{j-1}}(z)\backslash\{q\}$
		\EndFor
		\EndFor
		\For{$v\in V$ in parallel}
		\For{each message $(q,k,z)$}
		\State $p\gets$ any walk of length $2^{j-1}$ in $W_k(v)$ with end vertex $z$
		\State $W_k(v)\gets W_k(v)\backslash\{p\}\cup\{p+q\}$\Comment{$p+q$ is the concatenated walk}
		
		\EndFor
		\For{$k\equiv2^{j-1}\ (mod\ 2^j)$}
		\State $W_k(v)\gets\emptyset$
		\EndFor
		\EndFor
		\EndFor
		\State \Return $W_1(v)$ for all $v\in V$
		\EndProcedure
	\end{algorithmic}
\end{algorithm}

\begin{algorithm}
		\caption{Main Algorithm (Budgeting)}  
	\label{alg:main}
\begin{algorithmic}[1]
\Procedure{Main}{$G,r,\ell,B^*,\lambda$}\State $\theta\gets10C\ell^2\log n$,
	 $B_0\gets30C\lambda\ell^3\log n$,
	$\tau \gets 1+\sqrt{\frac{20C\log n}{\theta}}$ \Comment{\textbf{Parameter settings}}
	\State $\forall v \in V$, $B_0(v)\gets B_0 \cdot \mathrm{deg}(v)$
    \State $\forall v\in V,\ \forall k\in[\ell]:\ B(v,k)\gets B_0\cdot \mathrm{deg}(v)\cdot\tau^{3k-3}$
    \For{$i=1\ldots\lfloor \log_\lambda B^*\rfloor$}\label{line:cycle-for}\label{line:forlambda}
        \State $W_1\gets\textsc{Stitch(G,B)}$
        \State $W\gets W_1(r)$ \label{line:W}
        \For{$v\in V,\ k\in[\ell]$}
            \State $\kappa\gets|\{w\in W|w_k=v\}|$\label{line:kappa-def}
            \If{$\kappa\ge\theta$}\label{line:threshold}
                \State $B(v,k)\gets (B_0(v)+\lambda^i\cdot\frac{\kappa}{|W|})\cdot\tau^{3k-3}$\label{line:option-1}
            \Else
                \State $B(v,k)\gets B_0(v)\cdot\tau^{3k-3}$\label{line:option-2}
            \EndIf
        \EndFor
    \EndFor
    \State $W_1\gets\textsc{Stitch}(G,B)$
    \State $W\gets W_1(r)$
    \State \Return $W$
\EndProcedure
\end{algorithmic}
\end{algorithm}

\newpage
\begin{remark}
	In the above pseudocode (\Cref{alg:main}) $\tau$ is a scaling parameter slightly greater then one. We augment all budgets $B(\cdot,k)$ by a factor $\tau^{3k-3}$, to insure that there are always slightly more walks with higher labels, and ensure that stitching succeeds with high probability.
\end{remark}

\textbf{Analysis.} We are now ready to present the main properties of our algorithm:

\begin{lemma}[Correctness and complexity]\label{lem:alg}
\Cref{alg:main} takes $O(\log\ell\cdot\log_\lambda B^*)$ rounds of MPC communication, \textsc{Stitch} terminates without failures with high probability, and the total amount of memory used for walks is
    $\sum_{v\in V}\sum_{k=1}^\ell B(v,k)=O(m\lambda\ell^4\log n+B^*\lambda \ell).$
\end{lemma}

\begin{proof}

Let $P^k$ be distribution of random walks of length $k$ starting from $r$. That is, the probability that such a walk ends up in $v\in V$ after $k$ steps is $P^k(v)$.

To bound the round complexity we note that each call of \textsc{Stitch} takes only $O(\log\ell)$ rounds of communication, and it is called $\lfloor\log_\lambda B^*\rfloor + 1$ times; this dominates the round complexity. Further rounds of communication are needed to update the budgets. However this can be done in parallel for each vertex, and thus takes only one round per iteration of the outer for-loop (\Cref{line:cycle-for}).

To prove that the algorithm fails with low probability, we must show the following crucial claim about the budgets $B(V,K)$. Recall that the ideal budget in the $i^\text{th}$ iteration would be $B(v,k)\approx B_0(v) + \lambda^i\cdot P^k(v)$. We show that in reality, the budgets do not deviate too much from this.

\begin{claim}\label{claim:no-fail}
After iteration $i$ of the outer for-loop (\Cref{line:cycle-for}) in \Cref{alg:main}, with high probability $B$ is set such that
$$\forall v\in V,\ k\in[\ell]:\ B(v,k)\in\left[(B_0(v)+\lambda^i\cdot P^k(v))\cdot\tau^{3k-4},(B_0(v)+\lambda^i\cdot P^k(v))\cdot\tau^{3k-2}\right].$$
\end{claim}
\begin{proof}
We first note how $B(r,1)$ -- the budget of walks starting at the root vertex -- evolves. For $v=r$ and $k=1$, $\kappa$ is always equal to $|W|$ -- since $r$ is the root vertex -- and greater than $\theta$. Therefore, $B(r,1)$ is set in \Cref{line:option-1} of \Cref{alg:main} to $(B_0(r)+\lambda^i)$. This is important, because it means that when setting other budgets in iteration $i>1$, $|W|$ is always $(B_0(r)+\lambda^{i-1})$, the number of walks rooted from $r$ in the previous round. The exception is the first iteration, when $|W|$ is simply $|B_0(r)|$. In both cases we may say that $|W|\ge\lambda^{i-1}$.

There are two options we have to consider: If after the $i^\text{th}$ round of \textsc{Stitch} $\kappa$ exceeded $\theta$, in which case our empirical estimator $\kappa/|W|$ for $P^k(v)$ is deemed reliable. We then use this estimater to set the budget for the next round (see~\Cref{line:option-1}). Alternately, if $\kappa$ did not exceed $\theta$, the imperical estimator is deemed too unreliable; we then simply set $B(v,k)$ proportionally to $B_0(v)$ (see~\Cref{line:option-2}).

{\bf Case I ($\kappa < \theta$)} then intuitively, $\kappa$ is too small to provide an accurate estimator of $P^k(v)$. In this case we are forced to argue that the (predictable) term $B_0(v)$ dominates the (unknown) term $P^k(v)$. Since $\kappa<\theta$, $\mathbb E(\kappa)=P^k(v)\cdot|W|\le2\theta$. (The opposite happens with low probability\footnote{Throughout the proof, we say 'low probability' to mean probability of $n^{-\Omega(C)}$ where $C$ can be set arbitrarily high.} by Chernoff bounds, since $\theta\ge10C\log n$.) Therefore,
\begin{align*}
    B(v,k)&=B_0(v)\cdot\tau^{3k-3}\le(B_0(v)+\lambda^i\cdot P^k(v))\cdot\tau^{3k-2},
\end{align*}
and
\begin{align*}
    B(v,k)=B_0(v)\cdot\tau^{3k-3}=(B_0(v)+\lambda^i\cdot P^k(v))\cdot\left(1-\frac{\lambda^i\cdot P^k(v)}{B_0(v)+\lambda^i\cdot P^k(v)}\right)\cdot\tau^{3k-3}.
\end{align*}
So, we need to prove that $1-\frac{\lambda^i\cdot P^k(v)}{B_0(v)+\lambda^i\cdot P^k(v)}\ge\tau^{-1}$. Now, by the above bound on $\mathbb E(\kappa)$ as well as the fact that $|W|\ge\lambda^{i-1}$, we have
$
    2\theta\ge P^k(v)\cdot|W|\ge P^k(v)\cdot\lambda^{i-1},
$
which results in
$   \lambda^i\cdot P^k(v)\le2\lambda\theta.$
Consequently,
\begin{align*}
    \left(1-\frac{\lambda^i\cdot P^k(v)}{B_0(v)+\lambda^i\cdot P^k(v)}\right)\ge\left(1-\frac{2\lambda\theta}{B_0(v)+2\lambda\theta}\right)\ge\tau^{-1}.
\end{align*}
Here we used $B_0(v)\ge B_0\ge3\lambda\theta\cdot(\sqrt{\theta/(20C\log n)})$, which holds by definition of $B_0(v)$ and our setting of parameters $B_0$ and $\theta$.

{\bf Case II ($\kappa \ge \theta$)}, then intuitively $\kappa$ is robust enough to provide a reliable estimator for $P^k(v)$. More precisely, $\kappa/|W|\in\left[\mathbb E(\kappa/|W|)\cdot\tau^{-1},\mathbb E(\kappa/|W|)\cdot\tau\right]$ with high probability -- indeed $\tau$ is defined in terms of $\theta$ deliberately in exactly such a way that this is guaranteed by Chernoff bounds. $\mathbb E(\kappa/|W|)=P^k(v)$, therefore
\begin{align*}
    \lambda^i\cdot\frac{\kappa}{|W|}&\in\left[\lambda^i\cdot P^k(v)\cdot\tau^{3k-4},\lambda^i\cdot P^k(v)\cdot\tau^{3k-2}\right],
\end{align*}
and
\begin{align*}
    B(v,k)&=(B_0(v)+\lambda^i\cdot\frac{\kappa}{|W|})\cdot\tau^{3k-3}\\
    &\in\left[\left(B_0(v)+\lambda^i\cdot P^k(v)\right)\cdot\tau^{3k-4},(B_0(v)+\lambda^i\cdot P^k(v))\cdot\tau^{3k-2}\right].
\end{align*}

\end{proof}

\textsc{Stitch} only reports failure if for some $v\in V$ and $k\in[2,\ell]$, vertex $v$ receives more requests for walks in $W_k(v)$ than $|W_k(v)|=B(v,k)$. Number of such request is upper bounded by the number of edges ending in $v$ generated by neighbors of $v$, say $w$ at level $k-1$. That is, the number of requests for $W_k(v)$ is in expectation at most
\begin{align*}
    \sum_{w\in\Gamma(v)}\frac{1}{d(w)}B(w,k-1)&\le\sum_{w\in\Gamma(v)}\frac{1}{d(w)}(B_0(w)+\lambda^i\cdot P^{k-1}(w))\cdot\tau^{3k-5}\\
    &=\left(\sum_{w\in\Gamma(v)}\frac{1}{d(w)}B_0(w)+\lambda^i\cdot P^k(v)\right)\cdot\tau^{3k-5}\\
    &=\left(d(v)B_0+\lambda^i\cdot P^k(v)\right)\cdot\tau^{3k-5}\\
    &=\left(B_0(v)+\lambda^i\cdot P^k(v)\right)\cdot\tau^{3k-5}.
\end{align*}

Since this is greater than $\theta$, the actual number of requests is at most
$(B_0(v)+\lambda^i\cdot P^k(v))\cdot\tau^{3k-4}\le B(v,k),$
with high probability by Chernoff. Therefore, \textsc{Stitch} indeed does not fail.

Finally, we prove the memory bound. By setting of parameter $\theta$, $\tau^{3k-2}$ is at most a constant. Also by the setting of parameters we have,
\begin{align*}
    B(v,k)&\le(B_0(v)+\lambda^{\lfloor \log_{\lambda}B^*\rfloor + 1}\cdot P^k(v))\cdot\tau^{3k-2}=O(B_0(v)+B^*\lambda\cdot P^k(v)),
\end{align*}
and
$
    \sum_{v\in V}\sum_{k=1}^\ell B(v,k)=O(1)\cdot\sum_{v\in V}\sum_{k=1}^\ell(B_0(v)+B^*\lambda\cdot P^k(v))=O(m\lambda\ell^4\log n+B^*\lambda\ell).
$
\end{proof}
This gives us the proof of \Cref{thm:main}. Now, one can easily extend this result to the case when multiple sources for the starting vertex is considered.~\Cref{thm:general} is proven in~\Cref{sec:general}.

\section{From random walks to local clustering}\label{sec:PRclustering}
We now present two applications for our algorithm to compute random walks. In particular we show how to use it to compute PageRank vectors and how to use it to compute local clustering. In interest of space, we only state here our main results and we defer all the technical definition and proofs to the Appendix.
 
\textbf{Approximating PageRank using MPC random walks} Interestingly, we can show that we can use our algorithm as a primitive to compute PersonalizedPageRank for a node of for any input vector\footnote{We note that \cite{bahmani2011fast} also propose an algorithm to compute PersonalizedPageRank vector but with the limitation that this could be computed only for a single node and not for a vector.}

  \begin{theorem}[Approximating PersonalizedPageRank using MPC random walks]\label{thm:MPCapproxpr}
  	For any starting single vertex vector $s$ (indicator vector), any $\alpha \in (0,1)$ and any $\eta$, there is a MPC algorithm that  using $O(\log\ell\cdot\log_\lambda B^*)$ rounds of communication and the total amount of memory of
  	$O(m\lambda\ell^4\log n+B^*\lambda\ell)$, outputs a vector $\wt{q}$, such that $\wt{q}$ is a $\eta$-additive approximation to $\prr_{\alpha}(s)$, where $B^* := \frac{10^6\log^3 n}{\eta^2 \alpha^2}$ and $\ell:=\frac{10\log n}{\alpha}$.
  \end{theorem}
The proof is deferred to \Cref{sec:omitted}.

\textbf{Using approximate PersonalizedPageRank vectors to find sparse cuts}\label{sec:sparse-cuts}
Now we can use the previous result on PersonalizedPageRank to find sets with relatively sparse cuts. Roughly speaking, we argue that for any set $C$ of conductance $O(\alpha)$, for many vertices $v \in C$, if we calculate an approximate $\prr_\alpha(v)$ using our algorithms and perform a sweep cut over it, we can find a set of conductance $O(\sqrt{\alpha \log(\Vol(C))})$. This result is stated in \Cref{thm:final}. The proof of this result is very similar to the proofs of Section~5 of \cite{ACL}, however since our approximation guarantees are slightly different, we need to modify some parts of the proof for completeness. The full proof is presented in Appendix Our main result of this subsection is stated below.


\input{experiments}

\section*{Acknowledgments and Disclosure of Funding}
	This project has received funding from the European Research Council (ERC) under the European Union's Horizon 2020 research and innovation programme (grant agreement No 759471).
	
\newpage

\bibliographystyle{alpha}
\bibliography{references}

\newpage

\appendix

\section{Proof of~\Cref{thm:general}}\label{sec:general}

\general*

\begin{proof}
	
	First we consider the setting where all budgets $b_u$ are either the same value $b$, or $0$. We call vertices $u$, where $b_u=b$ roots, and the set of roots $R$. We can now run~\Cref{alg:main}, with two simple modification: In~\Cref{line:W} we set $W$ to be all rooted walks, that is $W\gets\cup_{r\in R}W_1(R)$. Correspondingly, in~\Cref{line:option-1}, we set the budget to $B(v,k)=(B_0(v)+R\cdot\lambda^i\cdot\frac{\kappa}{|W|})\cdot\tau^{3k-3}$, since there are now $R$ times as many rooted walks.
	
	From here, the proof of correctness proceeds nearly identically. In the case of a single vertex, we defined $P^k(v)$ as the probability that a walk from $r$ reaches $v$ as its $k^\text{th}$ step. Here we must define such a quantity for each $r\in R$: $P_r^k(v)$. The analogous claim to the central~\Cref{claim:no-fail} is that for all $v\in V$ and $k\in[\ell]$: $$B(v,k)\in\left[\left(B_0(v)+\lambda^i\cdot\sum_{r\in R} P^k_r(v)\right)\cdot\tau^{3k-4},\left(B_0(v)+\lambda^i\cdot\sum_{r\in R} P^k_r(v)\right)\cdot\tau^{3k-2}\right].$$
	
	In order to generalize to an arbitrary vector of budgets $(b_u)_{u\in V}$, we simply write $b$ as the summation of vectors $b^{(1)},\ldots,b^{(\log B^*)}$, where each vector $b_i$ has all of it's non-zero entries within a factor $2$ of each other. We then simply augment the coordinates of each $b_i$ where necessary, to get vectors $\wt b^{(i)}$ which have all non-zero entries equal to each other. At this point we have reverted to the simpler case: we can run our algorithm in parallel for all $\log B*$ budget vector, which incurs the insignificant extra factor of $\log B*$ in memory.
	
\end{proof}

\section{Preliminaries of \Cref{sec:PRclustering}}\label{sec:prelimPR}
For an undirected graph $G=(V,E)$, for each vertex $v\in V$, we denote its degree by $d(v)$ and for any set $S\subset V$, we define $\Vol(S) := \sum_{v\in S}d(v)$ and $\Vol(G) = 2|E|$. We define the stationary distribution over the graph as
\begin{align*}
	\forall v\in V:~~\psi(v) := \frac{d(v)}{\Vol(G)}
\end{align*}
For any vector $p$ over the vertices and any $S \subseteq V$ we define 
\begin{align*}
	p(S):=\sum_{v\in S} p(v).
\end{align*}
Moreover for any vector $p$ over the vertices, we define $p_+$ as follows:
\begin{align*}
	\forall v\in V:~~p_+(v) = \max(p(v),0).
\end{align*}

The \emph{edge boundary} of a set $S \subseteq V$ is defined as
\begin{align*}
	\partial(S) = \{  \{u,v\} \in E \text{ such that } u \in S, v \notin S   \}.
\end{align*}
The conductance of any set $S\subseteq V$ is defined as
\begin{align*}
	\Phi(S) = \frac{|\partial(S)|}{\min\{    \Vol(S), 2m-\Vol(S)\}}
\end{align*}

\paragraph{PageRank}
In the literature, PageRank was introduced for the first time in \cite{DBLP:journals/cn/BrinP98,ilprints422} for search ranking with starting vector of $s=\vec{\mathbf{1}}/n$ (the uniform vector). Later, personalized PageRank introduced where the starting vector is not the uniform vector, in order to address personalized search ranking problem and context sensitive-search \cite{Berkhin07bookmark-coloringapproach,DBLP:conf/waw/FogarasR04,1208999,10.1145/775152.775191}. In the rest of this paper we mostly work with personalized PageRanks, where the starting vector is an indicator vector for a vertex in the graph, and we use the general term of PageRank (as opposed to personalized PageRank) to avoid repetition. 
\begin{definition}[PageRank]
	The PageRank vector $pr_\alpha(s)$ is defined as the unique solution of the linear system
	$
	\prr_\alpha(s) = \alpha s + (1-\alpha)\prr_\alpha(s)W,
	$
	where $\alpha \in (0,1]$ and called the \emph{teleport probability}, $s$ is the \emph{starting vector}, and $W$ is the lazy random walk transition matrix $W:=\frac{1}{2}(I+D^{-1}A)$.
\end{definition} Below, we mention a few facts about PageRank vectors. 

\begin{fact}
	For any starting vector $s$, and any constant $\alpha \in (0,1]$, there is a unique vector $\prr_\alpha(s)$ satisfying $\prr_\alpha(s) = \alpha s + (1-\alpha)\prr_\alpha(s)W$.
\end{fact}
\begin{fact}\label{fact:PR}
	A PageRank vector is a weighted average of lazy random walk vectors. More specifically,
	$\prr_\alpha(s) = \alpha s + \alpha \sum_{t=1}^{\infty}(1-\alpha)^t(sW^t)$.
\end{fact}
Now, we define a notion of approximation that will be used throughout the paper. 
\begin{definition}\label{def:additive}($\eta$-additive approximations)
	We call a vector $q$, an $\eta$-additive approximate PageRank vector for $p:=\prr_\alpha(s)$, if for all $v\in V$, we have $q(v) \in \left[p(v) - \eta, p(v)+\eta\right]$.  
\end{definition}

\paragraph{Sweeps}
Suppose that we are given a vector $p$ that imposes an ordering over the vertices of graph $G=(V,E)$, as $v_1,v_2,\ldots,v_n$, where the ordering is such that 
\begin{align*}
	\frac{p(v_1)}{d(v_1)} \ge \ldots \ge \frac{p(v_n)}{d(v_n)}.
\end{align*} For any $j\in [n]$ define, $S_j:=\{  v_1,\ldots,v_j \}$. We define 
\begin{align*}
	\Phi(p) := \min_{i\in[n]} \Phi(S_i).
\end{align*}

\paragraph{Empirical vectors} Suppose that a distribution over vertices of the graph is given by a vector $q$. Now, imagine that at each step, one samples a vertex according to $q$, independently, and repeats this procedure for $M$ rounds. Let vector $N$ be such that for any vertex $v \in V$, $N(v)$ is equal to the number of times vertex $v$ is sampled. We call vector $\wt{q}$ a $(M,q)$-empirical vector, where
\begin{align*}
	\forall v \in V:~~\wt{q}(v):= \frac{N(v)}{M}
\end{align*}
\begin{claim}[Additive guarantees for empirical vectors]\label{claim:emprical}
	Let $q$ be a distribution vector over vertices of graph, where for each coordinate. Then, let vector $\wt{q}$ be a $(\frac{100}{\beta^2}\log n, q)$-empirical vector, for some $\beta$. Then $
	\forall v\in V: ~|q(v)-\wt{q}(v)| \le \beta $
	with high probability. 
\end{claim}
\begin{proof}
	Using additive Chernoff Bound (\cref{lem:additive-chernoff} with $N=\frac{100\log n}{\beta^2}$ and $\Delta = \beta$), for any $v\in V$, we have
	\begin{align*}
		\Pr[|q(v)-\wt{q}(v)| > \beta ] \le 2\exp\left(-2\frac{100\log n}{\beta^2}\beta^2\right)\le n^{-20}.
	\end{align*}
	Taking union bound over the vertices of the graph concludes the proof. 
\end{proof}

\section{Omitted claims, proofs and figures}\label{sec:omitted}

\begin{lemma}[Additive Chernoff Bound]\label{lem:additive-chernoff}
Let $X_1,X_2,\ldots,X_N\in[0,1]$ be $N$ iid random variables, let $\bar{X}:=(\sum_{i=1}^{N} X_i)/N$, and let $\mu = \mathbb{E}[\bar{X}]$. For any $\Delta>0$ we have 
\begin{align*}
	\Pr[\bar{X}-\mu \ge \Delta] \le \exp\left(-2N\Delta^2\right) 
\end{align*}
and
\begin{align*}
	\Pr[\bar{X}-\mu \le -\Delta] \le \exp\left(-2N\Delta^2\right).
\end{align*}
\end{lemma}

\begin{proofof}{\Cref{thm:MPCapproxpr}}
	We prove this theorem in a few steps. First, we prove that a proper truncation of the formula in \Cref{fact:PR} is a good approximation for PageRank vector:
	\begin{claim}\label{claim:truncated}
		For $T\ge \frac{10 \log n}{\alpha}$, we have that $q :=  \alpha s + \alpha \sum_{i=1}^{T}(1-\alpha)^i(sW^i)
		$ is a $n^{-10}$-additive approximate PageRank vector for $p:=\prr_\alpha(s)$. 
	\end{claim}
	\begin{proof}
		Since $s$ is an indicator vector and $W$ is a lazy random walk matrix, for any integer $t>0$, $sW^t$ is a distribution vector, and consequently every coordinate is bounded by $1$. So, for any vertex $v\in V$, we can bound $q(v)-p(v)$ in the following way:
		\begin{align*}
			|q(v)-p(v)| \le \alpha \sum_{i=T+1}^{\infty} (1-\alpha)^i
			\le (1-\alpha)^\frac{10 \log n}{\alpha} \le \left(e^{-\alpha}\right)^{\frac{10 \log n}{\alpha}} = n^{-10},
		\end{align*}
		since $1-\alpha \le e^{-\alpha}$ and $T \ge \frac{10\log n}{\alpha}$.
	\end{proof}
	From now on, we set $T := \frac{10\log n}{\alpha}$. Now, we show that using empirical vectors output by our parallel algorithm for generating random walks incurs small error. 
	\begin{claim}
		For any $i\in [T]$, let $q_i$ be the distribution vector for the end point of  lazy random walks of length $i$, output by the main algorithm with TVD error of $n^{-10}$ (see \Cref{thm:main}). Additionally, let vector $\wt{q_i}$ be a $(\frac{10^6\log^3 n}{\eta^2\alpha^2}, q_i)$-empirical vector. Now define 
		\begin{align*}
			\wt{q} := \alpha s + \alpha \sum_{i=1}^{T}(1-\alpha)^i \cdot \wt{q_i}
		\end{align*}
		for a constant $\alpha \in (0,1)$ and $T=\frac{10\log n}{\alpha}$. Then, $\wt{q}$ is an $\eta$-additive approximation to $p:=\prr_{\alpha}(s)$.
	\end{claim}
	\begin{proof}
		For the upper bound, for any $v\in V$ we have
		\begin{align*}
			\wt{q}(v) &=  \alpha s + \alpha \sum_{i=1}^{T}(1-\alpha)^i \cdot \wt{q_i}(v)\\
			&\le \alpha s + \alpha \sum_{i=1}^{T}(1-\alpha)^i \cdot \left(q_i(v)+\frac{\eta \alpha}{100\log n}\right)&&\text{By \cref{claim:emprical} with $\beta=\frac{\eta\alpha}{100\log n}$}\\
			&\le \alpha s +\alpha \sum_{i=1}^{T}(1-\alpha)^i \cdot q_i (v)+ \frac{\eta}{10}&&\text{Since $T=\frac{10\log n}{\alpha}$}\\
			&\le  \alpha s + \alpha \sum_{i=1}^{T}(1-\alpha)^i(sW^i)+ n^{-10}+ \frac{\eta}{10}&&\text{Using the main algorithm with TVD error $n^{-10}$}\\
			&\le p(v)+ 2n^{-10}+  \frac{\eta}{10}&&\text{By \cref{claim:truncated}}\\
			&\le p(v) + \eta.
		\end{align*}
		And similarly for the lower bound, for any $v\in V$ we have
		\begin{align*}
			\wt{q}(v) &=  \alpha s + \alpha \sum_{i=1}^{T}(1-\alpha)^i \cdot \wt{q_i}(v)\\
			&\ge \alpha s + \alpha \sum_{i=1}^{T}(1-\alpha)^i \cdot \left(q_i(v)-\frac{\eta \alpha}{100\log n}\right)&&\text{By \cref{claim:emprical} with $\beta=\frac{\eta\alpha}{100\log n}$}\\
			&\ge \alpha s +\alpha \sum_{i=1}^{T}(1-\alpha)^i \cdot q_i (v)- \frac{\eta}{10}&&\text{Since $T=\frac{10\log n}{\alpha}$}\\
			&\ge  \alpha s + \alpha \sum_{i=1}^{T}(1-\alpha)^i(sW^i)- n^{-10}+ \frac{\eta}{10}&&\text{Using the main algorithm with TVD error $n^{-10}$}\\
			&\ge p(v)- 2n^{-10}-  \frac{\eta}{10}&&\text{By \cref{claim:truncated}}\\
			&\ge p(v) - \eta.
		\end{align*}
	\end{proof}
	This means that we need to generate $B^*:=\frac{10^6\log^3 n}{\eta^2 \alpha^2}$ random walks of length $\ell := \frac{10\log n}{\alpha}$. Now, using \Cref{lem:alg} 
	\begin{enumerate}
		\item  in $O(\log\ell\cdot\log_\lambda B^*)$ rounds of MPC communication,
		\item and with the total amount of memory of
		$O(m\lambda\ell^4\log n+B^*\lambda \ell)$
	\end{enumerate}
	we can generate the required random walks. 
\end{proofof}
\begin{theorem}\label{thm:main1}
	Let $q$ be an $\eta$-additive approximate PageRank vector for $p:=\prr_\alpha(s)$, where $||s_+||_1\le 1$. If there exists a subset of vertices $S$ and a constant $\delta$ satisfying
	\begin{align*}
		q(S)-\psi(S) > \delta
	\end{align*}
	and $\eta$ is such that
	\begin{align*}
		\eta \le \frac{\delta}{8\left\lceil \frac{8}{\phi^2}\log(4\sqrt{\Vol(S)}/\delta)\right\rceil\min(\Vol(S),2m-\Vol(S))},
	\end{align*}
	then
	\begin{align*}
		\Phi(q) < \sqrt{\frac{18\alpha\log(4\sqrt{\Vol(S)}/\delta)}{\delta}}.
	\end{align*}
\end{theorem}
\begin{proofof}{\Cref{thm:main1}}
	Let $\phi := \Phi(q)$. By \Cref{lem:main}, for any subset of vertices $S$ and any integer $t$, we have
	\begin{align*}
		q(S)-\psi(S) \le \alpha t+\sqrt{X}\left(1-\frac{\phi^2}{8}\right)^t + 2t \cdot X\eta
	\end{align*}
	where $X:=\min(\Vol(S),2m-\Vol(S))$.
	If we set $$t = \left\lceil \frac{8}{\phi^2}\log(4\sqrt{\Vol(S)}/\delta)\right\rceil \le \frac{9}{\phi^2}\log(4\sqrt{\Vol(S)}/\delta),$$
	then we get
	\begin{align*}
		\sqrt{\min(\Vol(S),2m-\Vol(S))}\left(1-\frac{\phi^2}{8}\right)^t \le \frac{\delta}{4}.
	\end{align*}
	This results in
	\begin{align*}
		q(S)-\psi(S) \le \alpha \frac{9}{\phi^2}\log(4\sqrt{\Vol(S)}/\delta) +\frac{\delta}{4} + 2tX\eta
	\end{align*}
	Now, as we did set $\eta$ such that 
	\begin{align*}
		\eta \le \frac{\delta}{8tX}
	\end{align*}
	then since we assumed that $q(S)-\psi(S) \ge \delta$ then
	\begin{align*}
		\frac{\delta}{2} < \alpha \frac{9}{\phi^2}\log(4\sqrt{\Vol(S)}/\delta),
	\end{align*}
	which is equivalent to
	\begin{align*}
		\phi < \sqrt{\frac{18\alpha \log(4\sqrt{\Vol(S)}/\delta)}{\delta}}.
	\end{align*}
\end{proofof}
\begin{lemma}\label{lem:main}
	Let $q$ be an $\eta$-additive approximate PageRank vector for $p:=\prr_\alpha(s)$, where $||s_+||_1\le 1$. Let $\phi$ and $\gamma$ be any constants in $[0,1]$. Either the following bound holds for any set of vertices $S$ and any integer $t$:
	\begin{align*}
		q(S)-\psi(S) \le \gamma+\alpha t+\sqrt{X}\left(1-\frac{\phi^2}{8}\right)^t + 2t \cdot X\eta
	\end{align*}
	where $X:=\min\left(\Vol(S),2m-\Vol(S)\right)$,
	or else there exists a sweep cut $S_j^q$, for some $j\in [1,|\mathrm{Supp}(q)|]$, with the following properties:
	\begin{enumerate}
		\item $\Phi(S_j^q)<\phi$,
		\item For some integer $t$,
		\begin{align*}
			q(S_j^q)-\psi(S_j^q) > \gamma+\alpha t + \sqrt{X'}\left(1-\frac{\phi^2}{8}\right)^t+2t \cdot X'\eta,
		\end{align*}
		where $X':=\min(\Vol(S_j^q),2m-\Vol(S_j^q))$.
	\end{enumerate}
\end{lemma}
\begin{proofof}{\Cref{lem:main}}
	For simplicity of notation let $f_t(x):=\gamma+\alpha t + \sqrt{\min\left(x,2m-x\right)}\left(1-\frac{\phi^2}{8}\right)^t$. We are going to prove by induction that if there does not exist a sweep cut with both of the properties then equation 
	\begin{align}\label{eq:qxt}
		q[x]-\frac{x}{2m} \le f_t(x) +2t\cdot \min(x,2m-x)\eta
	\end{align}
	holds for all $t\ge 0$. 
	
	\paragraph{Base of induction ($t=0$):} We need to prove that for any $x\in[0,2m]$, $q[x] - \frac{x}{2m} \le \gamma + \sqrt{\min(x,2m-x)}$. The claim is true for $x \in [1,2m-1]$ since $q[x]\le 1$ for any $x$, so, we only need to prove the claim for $x\in [0,1]\cup [2m-1,2m]$. 
	
	\paragraph{Case I, $x\in[0,1]$:} For $x\in[0,1]$, $q[0] = 0$ and $q[1] \le 1$ and $q[x]$ is a linear function for $x\in[0,1]$. Also $\sqrt{\min(x,2m-x)}=\sqrt{x}$. Since $\sqrt{x}$ is a concave function then the claim holds for $x\in[0,1]$. 
	
	\paragraph{Case II, $x\in[2m-1,2m]$:} In this case $\sqrt{\min(2m-x,x)} + \frac{x}{2m}=\sqrt{2m-x} + \frac{x}{2m}$, which is a concave function. So we only need to check the end points of this interval. For $x= 2m$, the claim holds since $q[2m] = 1$. Similarly, for $x=2m-1$, $q[x] \le 1 \le \sqrt{1} + \frac{2m-1}{2m}$. 
	
	So the base of induction holds. 
	
	\paragraph{Inductive step:} Now assume that \Cref{eq:qxt} holds for some integer $t$. We prove that it holds for $t+1$. We only need to prove that it holds for $x_j = \Vol(S_j^q)$ for each $j\in [1,\mathrm{Supp}(q)]$.
	Consider any $j\in[1,|\mathrm{Supp}(q)]$, and let $S:=S_j^q$. If property 2 does not hold, then the claim holds. If property 1 does not hold, then we have $\Phi(S)\ge \phi$. Assume that $x_j \le m$ (the other case is similar)
	
	\begin{align*}
		q[\Vol(S)] - \frac{x_j}{2m} &= q(S) - \frac{x_j}{2m}  &&\text{Since $S$ is a sweep cut of $q$}\\
		&\le p(S) + |S|\cdot \eta - \frac{x_j}{2m} &&\text{By \Cref{def:additive}}
	\end{align*}  
	Let $F:=\mathrm{in}(S)\cap\mathrm{out}(S)$ and $F':=\mathrm{in}(S)\cup\mathrm{out}(S)$. By \Cref{lem:lazyPR}, 
	\begin{align}\label{eq:inout}
		p(S) = \alpha s(S) + (1-\alpha)\left(\frac{1}{2}p(F)+\frac{1}{2}p(F')\right).
	\end{align}
	Consequently, we have
	\begin{align*}
		q[x_j] &\le p(S)+|S|\cdot \eta\\
		&\le \alpha s(S) + (1-\alpha)\left(\frac{1}{2}p(F)+\frac{1}{2}p(F')\right) + |S|\cdot \eta &&\text{By \Cref{eq:inout}} \\
		&\le \alpha + \left(\frac{1}{2}p(F)+\frac{1}{2}p(F')\right) + |S|\cdot \eta &&\text{By $||s_+||_1\le1$ and $\alpha\in[0,1]$} \\
		&\le \alpha +\left(\frac{1}{2}q(F)+\frac{1}{2}q(F') + x_j\eta\right) + |S|\cdot \eta &&\text{By \Cref{claim:qpinout}}\\
		&\le \alpha  + \left(\frac{1}{2}q[x_j-|\partial(S)|]+\frac{1}{2}q[x_j + |\partial(S)|] + x_j\eta\right) + |S|\cdot \eta &&\text{By definition of $q[\cdot]$} \\
		&= \alpha  + \left(\frac{1}{2}q[x_j-\Phi(S)x_j]+\frac{1}{2}q[x_j + \Phi(S)x_j] + x_j\eta\right) + |S|\cdot \eta &&\text{By definition of $\Phi(S)$}\\
		&\le \alpha + \frac{1}{2}q[x_j-\phi x_j]+\frac{1}{2}q[x_j + \phi x_j] + 2x_j\eta &&\text{By concavity of $q$} \\
		&\le \alpha + \frac{1}{2}f_t[x_j-\phi x_j]+\frac{1}{2}f_t[x_j + \phi x_j] + 2t x_j\eta + \frac{x_j}{2m}+ 2x_j\eta &&\text{By induction assumption} \\
	\end{align*}
	Therefore
	\begin{align*}
		&q[x_j] - \frac{x_j}{2m} \\
		&\le  \alpha + \frac{1}{2}f_t[x_j-\phi x_j]+\frac{1}{2}f_t[x_j + \phi x_j] + 2(t+1)x_j\eta\\
		&= \gamma + \alpha + \alpha t +\frac{1}{2}\left(\sqrt{x_j-\alpha x_j} + \sqrt{x_j+\alpha x_j}\right) \left(1-\frac{\phi^2}{8}\right)^t+2(t+1)x_j\eta \\
		&\le \gamma + \alpha (t+1) +\sqrt{x_j} \left(1-\frac{\phi^2}{8}\right)^{t+1}+2(t+1)x_j\eta
	\end{align*}
\end{proofof}

\begin{definition}
	For any vertex $u\in V$ and any $v$ in neighborhood of $u$, we define 
	\begin{align*}
		p(u,v)=\frac{p(u)}{d(u)}.
	\end{align*}
	Also, we replace each edge $(u,v)\in E$ with two directed edges $(u,v)$ and $(v,u)$. Now, for any subset of directed edges $A$, we define \begin{align*}P(A)=\sum_{(u,v)\in A}p(u,v).\end{align*}
\end{definition}
\begin{definition}\label{def:inout}
	For any subset of vertices $S$, we define
	\begin{align*}
		\mathrm{in}(S) = \{(u,v)\in E|v\in S\}
	\end{align*}
	and 
	\begin{align*}
		\mathrm{out}(S) = \{(u,v)\in E|u\in S\}
	\end{align*}
\end{definition}
\begin{lemma}\label{lem:lazyPR}
	If $p=\mathrm{pr}_\alpha(s)$ is a PageRank vector, then for any subset of vertices $S$,
	\begin{align*}
		p(S) = \alpha(S) + (1-\alpha)\left(\frac{1}{2}p(\mathrm{in}(S)\cap\mathrm{out}(S))+\frac{1}{2}p(\mathrm{in}(S)\cup\mathrm{out}(S))\right).
	\end{align*}
\end{lemma}

\begin{claim}\label{claim:qpinout}
	Suppose that $q$ is an $\eta$-additive approximate PageRank vector for $p=\mathrm{pr}_\alpha(s)$ (see \Cref{def:additive}). Then, for any subset of vertices $S$, if we let $F:=\mathrm{in}(S)\cap\mathrm{out}(S)$ and $F':=\mathrm{in}(S)\cup\mathrm{out}(S)$,
	\begin{align*}
		-2\Vol(S)\eta \le \left(q(F)+q(F') \right)-\left( p(F)+p(F')\right) \le 2\Vol(S) \eta
	\end{align*}
\end{claim}
\begin{proof}
	By \Cref{def:inout}, if we define 
	\begin{align*}
		q(F)= \sum_{(u,v)\in F}\frac{q(u)}{d(u)} \le \sum_{(u,v)\in F} \frac{p(u)+\eta}{d(u)} \le \sum_{(u,v)\in F} \frac{p(u)}{d(u)} + \eta |F|=p(F) + \eta|F|.
	\end{align*}
	Similarly,
	\begin{align*}
		p(F) - \eta |F|\le q(F) .
	\end{align*}
	If we repeat the same procedure for $F':=\mathrm{in}(S)\cup\mathrm{out}(S)$, we get,
	\begin{align*}
		p(F') - \eta |F'| \le q(F')\le p(F') + \eta |F'|.
	\end{align*}
	In order to conclude the proof, we only need to note that 
	\begin{align*}
		|F|+|F'| = 2 \Vol(S).
	\end{align*}
\end{proof}

\begin{lemma}[Theorem 4 of \cite{ACL}]\label{lem:cut}
	For any set $C$ and any constant $\alpha \in (0,1]$, there is a subset $C_\alpha\subseteq C$ with volume $\Vol(C_\alpha) \ge \Vol(C)/2$ such that for any vertex $v\in C_\alpha$, the PageRank vector $\prr_\alpha(\chi_v)$ satisfies
	\begin{align*}
		[\prr_\alpha(\chi_v)](C) \ge 1-\frac{\Phi(C)}{\alpha}
	\end{align*}
where $[\prr_\alpha(\chi_v)](C)$ is the amount of probability from PageRank vector over set $C$. 
\end{lemma}
See \cite{ACL} for the proof of \Cref{lem:cut}.
\begin{lemma}\label{lem:approxcut}
	Let $\alpha \in (0,1]$ be a constant and let $C$ be a set satisfying
	\begin{enumerate}
		\item $\Phi(C) \le \alpha/10$,
		\item $\Vol(C) \le \frac{2}{3}\Vol(G)$.
	\end{enumerate}
	If $q$ is a $\eta$-additive approximation to $\prr_\alpha(\chi_v)$ where $v\in C_\alpha$ and $\eta\le  1/(10\Vol(C))$, then a sweep over $q$ produces a cut with conductance $\Phi(q) = O(\sqrt{\alpha\log(\Vol(C))})$.
\end{lemma}
\begin{proof}
	Since $q$ is a $\eta$-additive approximation to $\prr_\alpha(\chi_v)$, then using \Cref{lem:cut} we have 
	\begin{align*}
		q(C) &\ge 1-\frac{\Phi(C)}{\alpha}-\eta\cdot|C|
		\ge 1-\frac{\Phi(C)}{\alpha}-\eta\cdot\Vol(C),
	\end{align*}
	since $|C| \le \Vol(C)$. Combining this with the facts that $\Phi(C)/\alpha \le \frac{1}{10}$ and $\eta \le 1/(10\Vol(C))$, we have $q(C) \ge 4/5$, which implies
	\begin{align*}
		q(C)-\psi(C) \ge \frac{4}{5} - \frac{2}{3} = \frac{2}{15}.
	\end{align*}
	Now, \Cref{thm:main1} implies that 
	\begin{align*}
		\Phi(q) \le \sqrt{135\alpha \log(30\sqrt{\Vol(C)})}.
	\end{align*}
\end{proof}

\begin{proofof}{\Cref{thm:final}}
	The proof is by combining \Cref{thm:MPCapproxpr} and \Cref{lem:approxcut}.
\end{proofof}

\section{Additional Experiments}\label{sec:app-exp}

We present the result of experimentation with longer walks ($\ell=32$) in \Cref{table:exp-long}. Similarly to the other cases, the algorithm scales extremely well with the size of the graph. Furthermore, we observe that in the case of the smaller of the graphs (\textsc{com-DBLP}, \textsc{com-Youtube}), doubling the walk-length has a relatively small effect on the run-time. This is to be expected, as the number of Map-Reduce rounds performed scales logarithmically in $\ell$ (see \Cref{thm:main}). In the larger graphs, this is less evident, as the running time depends more and more on the work-load as opposed to the rounds complexity.

\begin{table}[H]
	\caption{Experiments with $\ell=32$, $C=3$, $B_0=5n/m$, $\lambda=32$, $\tau=1.3$.}
	\label{table:exp-long}
	
	\begin{center}
		\begin{small}
			\begin{sc}
				\begin{tabular}{lrrrr}
					\toprule
					Graph & \textbf{time} & $B_0$ & Rooted walks generated & Walk failure rate\\
					\midrule
					com-DBLP & \textbf{$25\pm2$ minutes} & 1.51 & $79,103\pm2412$ & $19.4\pm1.1\%$ \\
					com-Youtube & \textbf{$45\pm1$ minutes} & 1.9 & $44,839\pm179$ & $7.8\pm1\%$\\
					com-LiveJournal & \textbf{$115\pm3$ minutes} & 0.576 & $152,126\pm3028$ & $7.9\pm0.2\%$ \\
					com-Orkut & \textbf{$95\pm1$ minutes} & 0.131 & $163,056\pm1612$ & $5\pm0.1\%$ \\
					\bottomrule
				\end{tabular}
			\end{sc}
		\end{small}
	\end{center}
	\vskip -0.1in
\end{table}

In \Cref{table:exp-short} we see an experiment similar to that of \Cref{table:exp-rich}, but with the parameters $B_0$ and $\tau$ somewhat lowered. We confirm the results on \Cref{sec:scalablility} on the scaling of running time with the size of the graph. The lower parameters allow for faster running time. However, this is at the expense of both the walk failure rate and the number of rooted walks generated. With lower $B_0$ and $\tau$ the vertex budgets ($B(v,K)$ from \Cref{sec:main}) are smaller, and allow for higher relative deviation from the expectation, leading to more walk failure. The running time decrease is not significant, especially in the case of our smaller graphs, and we conclude that the setting of parameters in \Cref{table:exp-rich} are closer to optimal for most applications.

\begin{table}[H]
	\caption{Experiments with $\ell=16$, $C=3$, $B_0=3n/m$, $\lambda=32$, $\tau=1.2$.}
	\label{table:exp-short}
	
	\begin{center}
		\begin{small}
			\begin{sc}
				\begin{tabular}{lrrrr}
					\toprule
					Graph & \textbf{time} & $B_0$ & Rooted walks generated & Walk failure rate\\
					\midrule
					com-DBLP & \textbf{$17\pm1$ minutes} & 0.906 & $23,837\pm2210$ & $38.3\pm0.7\%$ \\
					com-Youtube & \textbf{$23\pm2$ minutes} & 1.14 & $15,977\pm2298$ & $28.1\pm1.7\%$\\
					com-LiveJournal & \textbf{$35\pm0$ minutes} & 0.346 & $57,460\pm2104$ & $26.2\pm0.5\%$ \\
					com-Orkut & \textbf{$33\pm1$ minutes} & 0.079 & $66,715\pm1502$ & $21.5\pm0.3\%$ \\
					\bottomrule
				\end{tabular}
			\end{sc}
		\end{small}
	\end{center}
	\vskip -0.1in
\end{table}

Finally, in \Cref{table:comp-large} we present the results of a comparison experiment, extremely similar to that of \Cref{table:comp-small}, but with $\lambda$ increased to $20$. The discrepancy is even more striking. Increasing the target budget by a factor of $4$ produces no measurable difference for \Cref{alg:main}. However, \textsc{Uniform Stitching} is no longer able to complete on the cluster for inputs \textsc{email-Enron} and \textsc{com-DBLP}, due to the high memory requirement (denoted as '---').

\begin{table}[H]
	\caption{Experiments with $\ell=16$, $\lambda=20$, $\tau=1.3$. The row labeled '\Cref{alg:main}' corresponds to $B_0=1$, $C=3$, while the row labeled 'Uniform Stitching' corresponds to $B_0=400$, $C=1$.}
	\label{table:comp-large}
	\begin{center}
		\begin{small}
			\begin{sc}
				\begin{tabular}{lrrr}
					\toprule
					\textbf{Algorithm} & ca-GrQc & email-Enron & com-DBLP \\
					\midrule
					\textbf{\Cref{alg:main}} & $15\pm1$ minutes & $19\pm1$ minutes & $17\pm1$ minutes \\
					\textbf{Uniform stitching} & $8\pm0$ minutes & --- & --- \\
					\bottomrule
				\end{tabular}
			\end{sc}
		\end{small}
	\end{center}
	\vskip -0.1in
\end{table}

\paragraph{Implementation details.} In \Cref{alg:main}, $B(v,k)$ --  the budget associated with the $k^\text{th}$ step of the random walk -- is proportional to $\tau^{3k}$ (see \Cref{line:option-1} and \Cref{line:option-2}) which can lead to a factor $\tau^{\theta(\ell)}$ blow-up in space. In theory this is not a significant loss asymptotically, due to the settings of $\tau$ and $\theta$. Nonetheless, in practice, we use a more subtle formula which leads only to a factor $\tau^{\log_2\ell}$ blow-up, while retaining a similar guarantee on the probability of failure.

Furthermore, in \Cref{alg:main} (and the intuitive explanation before it) we distinguish between $W_k(v)$ for different $k$. That is walk segments have predetermined positions in the walk, and a request to stitch to a walk ending in $v$ with its $k^\text{th}$ step can only be served by a walk starting in $v$ with its $k+1^\text{st}$ step. This is mostly for ease of understanding and analysis. In the implementation we make no such distinction. Each node simply stores a set of walks of length $2^i$ in the $i^\text{th}$ round. The initial budget of each vertex $v$ (at the beginning of the cycle) is set to $\sum_kB(v,k)$, where $B(v,k)$ is still calculated according to the formulas in \Cref{line:option-1} and \Cref{line:option-2} of \Cref{alg:main} (with the exception of the altered $\tau$-scaling term, as mentioned in the paragraph above).

\end{document}

%% file: experiments.tex

\section{Empirical Evaluation}\label{sec:experiments}

In this Section we present empirical evaluations of our algorithms for random walk generation, as well as clustering. As our datasets, we use several real-world graphs form the Stanford Network Analysis Project~\cite{snapnets, DBLP:journals/tkdd/LeskovecKF07, DBLP:journals/corr/abs-0810-1355, Klimt2004IntroducingTE, DBLP:journals/corr/abs-1205-6233}. The graphs are undirected and come mostly (though not in all cases) from 'Networks with ground truth communities', where the clustering application is most relevant. In order to demonstrate the scalability of our main Algorithm we use graphs of varying sizes, as demonstrated in the table below.

\begin{table}[H]
	\caption{Summary of the various graphs used in our empirical evaluations.}
	\label{table:graphs}

	\begin{center}
		\begin{small}
			\begin{sc}
				\begin{tabular}{lrrl}
					\toprule
					\textbf{Name} & Vertices & Edges & Description  \\
					\midrule
					\textbf{ca-GrQc} & 5424 & 14,496 & Collaboration network \\
					\textbf{email-Enron} & 36,692 & 183,831 & Email communication network \\
					\textbf{com-DBLP} & 317,080 & 1,049,866 & Collaboration network \\
					\textbf{com-Youtube} & 1,134,890 & 2,987,624 & Online social network \\
					\textbf{com-LiveJournal} & 3,997,962 & 34,681,189 & Online social network \\
					\textbf{com-Orkut} & 3,072,441 & 117,185,083 & Online social network \\
					\bottomrule
				\end{tabular}
			\end{sc}
		\end{small}
	\end{center}
	\vskip -0.1in
\end{table}

The experiments were performed on Amazon's Elastic Map-Reduce system using the Apache Hadoop library. The clusters consisted of 30 machines, each of modest memory and computing power (Amazon's \texttt{m4.large} instance) so as to best adhere to the MPC setting. Each experiment described in this section was repeated 3 times to minimize the variance in performance inherent in distributed systems like this.

\paragraph{Practical considerations.} In~\Cref{sec:main} we worked with the guarantee that no walks "fail" in the stitching phase. This assumption can be fulfilled at nearly no expense to the (asymptotic) guarantees in space and round complexity of~\Cref{thm:main}, and make the proof much cleaner. In practice however, allowing some small fraction of the walks to fail allows for a more relaxed setting of the parameters, and thus better performance.

Each experiment is performed with $15$ roots, selected uniformly at random. The main parameters defining the algorithm are as follows: $\ell$ --- the length of a target random walk (16 and 32 in various experiments), $C$ --- The number of cycles (iterations of the for-loop in \Cref{line:cycle-for} of \Cref{alg:main}) performed. 
, $B_0$ --- the initial budget-per-degree of each vertex, $\lambda$ --- the approximate scaling of the budgets of the root vertices each cycle, $\tau$ --- a parameter defining the amount of excess budget used in stitching. This is used somewhat differently here than in \Cref{alg:main}. For more details see \Cref{sec:app-exp}.


\subsection{Scalability}\label{sec:scalablility}

In this section we present the results of our experiments locally generating random walks simultaneously from multiple root vertices. We use the graphs \textsc{com-DBLP}, \textsc{com-Youtube}, \textsc{com-LiveJournal}, and \textsc{com-Orkut}, in order to observe how the runtime of~\Cref{alg:main} scales with the size of the input graph. In each of these graphs, $15$ root vertices have been randomly chosen. We ran three experiments with various settings of the parameters, of which one is presented below. For additional experiments see \Cref{sec:app-exp}. $B_0$ is set to be proportional to $n/m$ -- that is inverse proportional to the average degree -- since the initial budget of each of each vertex is set to $B_0$ times the degree of the vertex.

We report the execution time in the Amazon Elastic Map-Reduce cluster, as well as the number of rooted walks generated. Finally, under 'Walk failure rate', we report the percentage of \emph{rooted} walks that failed in the \emph{last cycle} of stitching. This is the crucial quantity; earlier cycles are used only to calibrate the vertex budgets for the final cycle.\footnote{For completeness, we report the empirical standard deviation everywhere. Note, however, that due to the high resource requirement of these experiments, each one was only repeated three times.} 

\begin{table}[H]
	\caption{Experiments with $\ell=16$, $C=3$, $B_0=6n/m$, $\lambda=32$, $\tau=1.4$.}
	\label{table:exp-rich}
	
	\begin{center}
		\begin{small}
			\begin{sc}
				\begin{tabular}{lrrrr}
					\toprule
					Graph & \textbf{time} & $B_0$ & Rooted walks generated & Walk failure rate\\
					\midrule
					com-DBLP & \textbf{$23\pm7$ minutes} & 1.812 & $96,362\pm2597$ & $14.6\pm0.6\%$ \\
					com-Youtube & \textbf{$34\pm6$ minutes} & 2.279 & $53,076\pm1185$ & $10.8\pm0.5\%$\\
					com-LiveJournal & \textbf{$76\pm11$ minutes} & 0.692 & $184,246\pm756$ & $7.9\pm0.1\%$ \\
					com-Orkut & \textbf{$64\pm13$ minutes} & 0.157 & $200,924\pm1472$ & $3.4\pm0.0\%$ \\
					\bottomrule
				\end{tabular}
			\end{sc}
		\end{small}
	\end{center}
	\vskip -0.1in
\end{table}

We observe that~\Cref{alg:main} successfully generates a large number of rooted walks -- far more than the initial budgets of the root vertices. As predicted, execution time scales highly sublinearly with the size of the input (recall for example, that \textsc{com-Orkut} is more than a hundred times larger than \textsc{com-DBLP}). The failure rate of walks decreases with the size of the graph in this dataset, with that of \textsc{com-Orkut} reaching as low as $3.4\%$ on average; this may be due to the the higher average degree of our larger graphs, leading to the random walks spreading out more.

In \Cref{sec:app-exp} we report the results of two more experiments, including one with longer walks.

\subsection{Comparison}\label{sec:comparison}

In this section we compare to the previous work of~\cite{lkacki2020walking} for generating random walks in the MPC model. This work heavily relies upon generating random walks from all vertices simultaneously, with the number of walks starting from a given vertex $v$ being proportional $d(v)$. In many applications, however, we are interested in computing random walks from a small number of root vertices. The only way to implement this using the methods of~\cite{lkacki2020walking} is to start with an initial budget large enough to guarantee the desired number of walks from each vertex.

We perform similar experiments to those in the previous section -- albeit on much smaller graphs. Each graph has $15$ root vertices chosen randomly, from which we wish to sample random walks. In~\Cref{table:comp-small} we set $B_0$ to $1$ and perform $C=3$ cycles of~\Cref{alg:main} with $\lambda=10$, effectively augmenting the budget of root vertices by a factor $100$ by the last cycle. Correspondingly, we implement the algorithm of~\cite{lkacki2020walking} -- which we call \textsc{Uniform Stitching} -- by simply setting the initial budget $100$ times higher, and performing only a single cycle of stitching, ie.: $B_0=100$, $C=1$.

\begin{table}[H]
	\caption{Experiments with $\ell=16$, $\lambda=10$, $\tau=1.3$. The row labeled~\Cref{alg:main}' corresponds to $B_0=1$, $C=3$, while the row labeled 'Uniform Stitching' corresponds to $B_0=100$, $C=1$.}
	\label{table:comp-small}
	\begin{center}
		\begin{small}
			\begin{sc}
				\begin{tabular}{lrrr}
					\toprule
					\textbf{Algorithm} & ca-GrQc & email-Enron & com-DBLP \\
					\midrule
					\textbf{\Cref{alg:main}} & $15\pm1$ minutes & $19\pm1$ minutes & $18\pm1$ minutes \\
					\textbf{Uniform stitching} & $7\pm0$ minutes & $15\pm0$ minutes & $66\pm1$ minutes \\
					\bottomrule
				\end{tabular}
			\end{sc}
		\end{small}
	\end{center}
	\vskip -0.1in
\end{table}

We observe that the running time of our \Cref{alg:main} barely differs across the three graphs, despite the nearly $100$-factor difference between the sizes of \textsc{ca-GrQc} and \textsc{com-DBLP}. At this small size, the execution time is dominated by setting up the $12$ Map-Reduce rounds required to execute~\Cref{alg:main} with these parameters. As expected, the baseline far outperforms our algorithm on the smallest graph; as the size of the input graph grows, however, its running time deteriorates quickly. For larger setting of the parameters \textsc{UniformStitching} can no longer complete on the cluster, due to the higher memory requirement, as we can see in~\Cref{table:comp-large} in~\Cref{sec:app-exp}.